\documentclass[reqno]{amsart}
\usepackage{graphicx} \usepackage{mathtools,amssymb}
\usepackage{xcolor}
\usepackage[normalem]{ulem}
\usepackage{tikz}
\usepackage[T1]{fontenc}
\usetikzlibrary{arrows.meta,positioning,decorations.pathreplacing,shapes.callouts, calc}
\usepackage[colorlinks=true, linkcolor=blue, citecolor=red, urlcolor=magenta]{hyperref}
\usepackage{thm-restate}
\usepackage{enumerate}

\theoremstyle{plain}
\newtheorem{thm}{Theorem}[section]
\newtheorem{lem}[thm]{Lemma}
\newtheorem{prop}[thm]{Proposition}

\newtheorem{cor}[thm]{Corollary}
\newtheorem{rmk}[thm]{Remark}

\theoremstyle{definition}
\newtheorem{defn}[thm]{Definition}

\renewcommand{\leq}{\leqslant}
\renewcommand{\geq}{\geqslant}
\renewcommand{\le}{\leqslant}
\renewcommand{\ge}{\geqslant}  \def\PP{{\mathbb{P}}}

\def\F{{\mathbb{F}}}

\def\Fq{{\mathbb{F}_q}}
\def\V{{\mathsf{V}}}
\def\T{{\mathsf{T}}}

\def\E{{\mathcal{E}}}
\def\H{{\mathcal{H}}}
\def\P{{\mathcal{P}}}
\def\Q{{\mathcal{Q}}}

\def\prm{{\mathsf{PRM}}}
\def\Ev{{\rm{Ev}}}
\def\N{{\mathcal{N}}}

\def\L{{\mathcal{L}}}
\def\Fqbar{{\overline{\mathbb{F}}_q}}
\def\Fqbartimes{{\overline{\mathbb{F}}^{\times}_q}}
\def\FqbarN{{\overline{\mathbb{F}}^{N+1}_q}}
\def\FqN{{\mathbb{F}^{N+1}_q}}
\def\sing{{\rm{Sing}\hspace{0.5mm} }}

\def\rk{{\rm{Rk}\hspace{0.5mm} }}
\def\rad{{\rm{Rad}\hspace{0.5mm} }}

\def\supp{{\rm{Supp}\hspace{0.5mm}}}
\def\wH{{\rm{w_H}\hspace{0.5mm}}}
\def\im{{\rm{im}\hspace{0.5mm}}}

\newcommand{\eqdef}{\stackrel{\text{def}}{=}}
\newcommand{\map}[4]{
  \left\{
  \begin{array}{ccc}
    #1 & \longrightarrow & #2 \\
    #3 & \longmapsto     & #4
  \end{array}
  \right.
}

\newcommand{\GaussBinom}[2]{{\left[
      \begin{array}{c}
        #1 \\ #2
      \end{array}
\right]}_q}

\title[Maximal quadrics and minimal codewords]{Maximal quadrics over finite fields and minimal codewords of projective Reed--Muller codes}
\author{Alain Couvreur}
\author{Rati Ludhani}
\thanks{\tiny{E--mail address: \texttt{\{alain.couvreur, rati.ludhani\}@inria.fr}\\
Inria \& Laboratoire LIX, École Polytechnique, Institut Polytechnique de Paris, 91120 Palaiseau CEDEX, France.}}

\date{}

\begin{document}

\begin{abstract}
  We study the classification of minimal codewords of projective Reed--Muller codes of order $2$. This problem is equivalent to identifying quadrics over finite fields whose set of rational points is maximal with respect to the inclusion. We prove that except one particular case over $\F_2$, any two absolutely irreducible quadrics whose sets of rational points are contained within one another should be equal as projective varieties. We deduce a precise characterisation of the minimal codewords of projective Reed--Muller codes of order $2$ and further give their exact number for each possible weight. 
\end{abstract}

\maketitle



\section{Introduction}
Minimal codewords of a linear code carry important information about its combinatorial and geometric structure. A nonzero codeword is said to be \emph{minimal} if its support does not properly contain the support of any other nonzero codeword. Massey~\cite{massey1993minimal} proved that the access structure of a linear secret sharing scheme can be completely determined by the minimal codewords of the dual of the code considered for the scheme. Some decoding algorithms are also based on minimal codewords over the binary field, see~\cite{hwang1979decoding, agrell2002voronoi1,agrell2002voronoi2,ashikhmin2002minimal}. Further, just over a decade ago, Johnsen and Verdure~\cite{johnsen2013hamming} showed that one can associate a vector matroid to a linear code, where the minimal codewords correspond precisely to the circuits of this matroid. These circuits play a crucial role in determining the weight distribution and higher weight spectra of the code.  Thus, identifying minimal codewords is of significant interest, although doing so in practice is highly nontrivial. 
A general algorithm is known for binary codes \cite{agrell2002voronoi2}, yet it quickly becomes inefficient for codes of large size. This has led to studies focusing on particular code families where additional structure can be exploited; see, for instance~\cite{agrell2002voronoi2, ashikhmin2002minimal, borissov2004minimal, schillewaert2010minimal}. In this work, we address this problem for projective Reed–Muller codes of order 2.

Let $q$ be a prime power and let $\Fq$ denote the finite field with $q$ elements. Write $\Fqbar$ for an algebraic closure of $\Fq$. 
Let $N$ and $d$ be nonnegative integers. We write $\PP^N$ for the $N$-dimensional projective space over $\Fqbar$, that is, $(\overline{\F}_q^{N+1} \setminus \{0\})/{\Fqbartimes}$.
Let $\PP^N(\Fq)$ be its set of rational points and write \[p_N\eqdef \sharp \PP^N(\Fq).\] Denote by $\Fq[X_0,\ldots,X_N]$ the polynomial ring in $N+1$ variables with coefficients in $\Fq$ and by ${\Fq[X_0,\ldots,X_N]}_d$ the $\Fq$-space of homogeneous polynomials of degree $d$, including the zero polynomial. A \emph{projective Reed--Muller} (PRM) code \emph{of order $d$} and length $p_N$, denoted by $\prm_q(d,N)$, is defined as the image of the evaluation map \[
  \Ev : \map{{\Fq[X_0,X_1,\ldots,X_N]}_{d}}{\F^{^{p_N}}_q}{f}{{(f(P))}_{P \in E}~,}
\]
where $E$ is a list of $p_N$ elements of $\FqN$ representing any element of $\PP^N(\Fq)$.

PRM codes were first mentioned as an example of geometric Goppa codes from higher dimensional varieties in the lecture notes~\cite{vladut1984linear} by Manin and
Vl\u{a}du{\c{t}}.
Later, Lachaud~\cite{lachaud1986projective} formally introduced them for order $d<q$ and S{\o}rensen~\cite{sorensen1991projective} generalized to $d\geq q$. Lachaud in~\cite[\S 3]{lachaud1990parameters} further proved that PRM codes behave better 
than the ubiquitous family of codes called Reed--Muller codes, in particular, the performance measure (relative distance--transmission rate sum) is higher for PRM codes compared to Reed-Muller codes in a range of parameters. Thus, it is important to understand PRM codes better to see their potential for applications. Several fundamental parameters of PRM codes are already known: the minimum distance was determined by Lachaud~\cite{lachaud1990parameters} using  Serre's bound~\cite{serre1989lettrea} for $d<q$ and independently by S{\o}rensen~\cite{sorensen1991projective} in general, the minimum weight codewords were characterized by Rolland~\cite{rolland2007number}, and their enumeration was obtained by Ghorpade and the second named author~\cite{ghorpade2023minimum}.

\subsection*{Our contributions} In this article, we characterize all minimal codewords of $\prm_q(2,N)$ and determine their exact number at each weight. As one can observe from the definition, these questions are equivalent to characterising the
quadrics over $\Fq$ whose sets of $\Fq$--rational points are maximal with respect to the inclusion and determining the number of such sets having same cardinality.

Following Hirschfeld and Thas~\cite{hirschfeld1991general}, absolutely irreducible quadrics in $\PP^N$ over $\Fq$ fall into three classes: parabolic, hyperbolic, and elliptic. Thus, to characterise quadrics with respect to their $\Fq$--rational points, our key ingredient would be to prove the following result which may also be of independent interest in finite geometry, combinatorics or number theory.

\begin{restatable}{thm}{main}\label{thm:main}
  Let $\Q$ and $\Q'$ be two absolutely irreducible quadrics in
  $\PP^{N}$ over $\Fq$. 
If $\Q(\Fq)\subset \Q'(\Fq)$, then $\Q=\Q'$, except when $q=2$, $\Q$ is elliptic and $\Q'$ is hyperbolic both of rank $4$.
\end{restatable}

\begin{restatable}{cor}{characterisation}\label{cor:characterisation}
  A quadric $\Q$ in $\PP^N$ defined over $\Fq$ is uniquely determined
  by its set of $\Fq$--rational points, unless it is irreducible and not
  absolutely irreducible (\emph{i.e.} if it is the union of two
  $\F_{q^2}$--hyperplanes that are conjugate under the action of the
  Frobenius map).
\end{restatable}

We emphasize that our approach handles both the odd and even characteristic cases of $\Fq$ in a unified way.

It is worth noting that an analogous result is easily seen to hold for linear spaces, \emph{i.e.} if two linear spaces $L_1$ and $L_2$ in $\PP^N$ over $\Fq$ are such that $L_1(\Fq)\subset L_2(\Fq)$, then $L_1\subset L_2$. This indicates that quadrics mirror some structural aspects of linear spaces. 
One possible approach toward proving Theorem~\ref{thm:main} is to invoke the Lang–Weil bound~\cite{lang1954number} which provides estimates on the number of $\Fq$--rational points of absolutely irreducible varieties. In this context, one could apply it for both $\Q$ and $\Q\cap \Q'$, whenever absolutely irreducible, to compare their $\Fq$--rational points. However, such bounds are sharp only for large $q$, and as $N$ grows, $q$ must grow rapidly for the bound to provide a contradiction.
Therefore, toward finer results over small fields we had to adopt another approach by combining together geometric and combinatorial techniques. Our main approach relies on induction on the ambient dimension $N$ and how rank behaves on hyperplane sections of quadrics, which provide an elegant framework. 

Our main theorem for minimal codewords of $\prm_q(2,N)$ is as follows.

\begin{restatable}{thm}{mincodewords}\label{thm:intro:mincodewords}
  The minimal codewords of the codes $\prm_q(2,N)$ are exactly the codewords $c=\Ev(F)$ where $F$ has one of the following forms: 
  	\begin{enumerate}
  			\item\label{item:mincodewords1} $L_1 L_2$  where $L_1$ and $L_2$ are $\Fq$-linearly independent homogeneous linear forms in $\Fq[X_0,\ldots,X_N]$; 
  			\item\label{item:mincodewords2} $F$ is an absolutely irreducible quadratic form in $\Fq[X_0,\ldots,X_N]$ excluding the forms of rank $3$ if $q\leq 3$ and of rank $4$ when elliptic if $q=2$.
	\end{enumerate}   
\end{restatable}

Note that Case~(\ref{item:mincodewords1}) in the above theorem
corresponds exactly to the minimum weight codewords of $\prm_q(2,N)$,
see \cite[Lem.~2.3]{rolland2007number}.

As a final contribution, we count the number of
minimal codewords using the classification from
Theorem~\ref{thm:intro:mincodewords} and following Hirschfeld and
Thas~\cite[Thm.~1.44(ii)]{hirschfeld1991general} for the number of smooth quadrics in each orbit with respect to the action of
$\mathbf{PGL}_N(\Fq)$.

\subsection*{Outline of the article}
In Section~\ref{sec:setup}, we develop the necessary background by recalling known results needed in the sequel and introducing several new lemmas required for our main theorems. In Section~\ref{sec:max_quadrics}, we prove Theorem~\ref{thm:main}. Finally, we enumerate in Section~\ref{sec:minimal_codewords} the minimal codewords of $\prm_q(2,N)$. 

\section*{Acknowledgements}
The authors are partially funded by French \emph{Agence Nationale de la
Recherche} through the \emph{France 2030} project ANR-22-PETQ-0008 PQ-TLS and
the \emph{projet de recherche collaboratif} ANR-
21-CE39-0009-BARRACUDA, and by Horizon-Europe MSCA-DN project
ENCODE. The second author has received funding from the European Union's
Horizon 2020 research and innovation program under the Marie
Sk{\l}odowska-Curie grant agreement n$^\circ$101034255.

\section{Setup and Auxiliary Results}\label{sec:setup}
Let $q$ be a power of a prime $p$ and $\Fq$ be the finite field with $q$ elements.

\subsection{Projective varieties}
Let $\Fqbar$ denote an algebraic closure of $\Fq$. Let $N$ be a positive integer. For any field $K$, let $K[X_0,\ldots,X_N]$ denote the polynomial ring in $N+1$ variables defined over $K$. For each nonnegative integer $d$, we denote by ${K[X_0,\ldots,X_N]}_d$ the $K$-vector space consisting of all homogeneous polynomials of degree $d$ including the zero  polynomial. Let $\PP^N(\Fqbar)$ or $\PP^N$ for short, denote the $N$-dimensional projective space over $\Fqbar$, defined as $({\overline{\F}_q^{N+1}\setminus \{\bf{0}\}})/{\Fqbartimes}$.
For any $F\in \Fq[X_0,\ldots,X_N]$, let $\V(F)$ denote the (projective) variety defined by $F$ in $\PP^N$, \emph{i.e.} the set of zeros of $F$ in $\PP^N(\Fqbar)$.

Take note that, in the present article, despite studying quadrics over
a non--algebraically closed field, we chose to keep a basic geometric
point of view and to avoid the formalism of schemes.  In short, any
geometric object: projective space, quadric hypersurfaces, linear
subspaces and any other variety, \textbf{should be understood as a subset of} $\PP^N(\Fqbar)$,
\emph{i.e.} as the whole set of $\Fqbar$--rational points satisfying
some polynomial equations. Hilbert's Nullstellensatz establishes the
one--to--one correspondence between such a geometric object and the
corresponding homogeneous radical ideal. A variety is said to be
\emph{defined over $\Fq$} if the corresponding ideal is generated by
elements of $\Fq[X_0, \dots, X_N]$ or equivalently, if the variety is
left globally invariant under the Galois action of
$\text{Gal}(\Fqbar/\Fq)$.

Denote by $\PP^N(\Fq)$ the set of $\Fq$-rational points of $\PP^N$ which has cardinality
\[p_N\eqdef \sharp \PP^N(\Fq) = q^N+\cdots+q+1.\]
Similarly, for any variety $\mathcal{X}$ in $\PP^N$, we denote by $\mathcal{X}(\Fq)$ the set of \emph{$\Fq$--rational points} of $\mathcal{X}$. Moreover, since we will only be concerned with the $\Fq$--rational points of a variety throughout this paper, we will refer to them as \emph{rational points}.

\subsection{Linear codes and projective Reed--Muller codes}\label{subsec:codes}

	An {\it $[n,k]_q$--(linear) code} is a $k$--dimensional $\Fq$--subspace of $\F^n_q$ and its elements are said to be the {\it codewords}. Let $C$ be an $[n,k]_q$-code. The {\it support and Hamming weight of a codeword $c=(c_1,\ldots,c_n)\in C$} are defined as
	\[
		\supp(c)\eqdef \{i\in \{1,\ldots,n\}: c_i\ne 0\} \quad \text{and} \quad \wH(c)\eqdef \sharp \supp(c).
	\]

A nonzero codeword $c$ of an $[n,k]_q$-code $C$ is said to be {\it minimal} if there does not exist any other nonzero codeword $c'$ such that $\supp(c')\subsetneq \supp(c)$.

\begin{defn}
	Let $d$ be a nonnegative integer. Choose the representatives of the rational points of $\PP^N$ in $\F^{N+1}_q$ with last nonzero coordinate $1$ and list them as $P_1,\ldots,P_{p_{{}_N}}$. Define the evaluation map 
\[
  \Ev : \map{{\Fq[X_0,\ldots,X_N]}_d}{\F^{^{p_{_N}}}_q}{F}{(F(P_1),\ldots,F(P_{p_{_{N}}})).}
\]
The image of $\Ev$ is a linear code called the \emph{projective Reed--Muller code of length $p_{{}_N}$ and of order $d$} and denoted by $\prm_q(d,N)$.
\end{defn} 

Following Manin and Vl\u{a}du\c{t}' works \cite{vladut1984linear},
Lachaud formally introduced and studied the parameters of projective Reed-Muller codes for $d<q$ in \cite{lachaud1986projective,lachaud1990parameters}. In particular, 
for $d<q$, the dimension of $\prm_q(d,N)$ is simply $m+d\choose d$ as $\Ev$ is injective and the minimum distance is due to Serre~\cite{serre1989lettrea}. For $d\geq q$, 
the dimension and minimum distance of $\prm_q(d,N)$ are known by S\o rensen~\cite{sorensen1991projective}. Several dimension formulas for general $d$ also appear in the literature. One may refer to article~\cite{ghorpade2023minimum} for a comprehensive exposition of the code parameters, including a discussion of the equivalence of these formulas.

It is evident that for $c_F\in \prm_q(d,N)$ where $F\in {\Fq[X_0,\ldots,X_N]}_d$, 
\[
\supp(c_F)=\{i\in \{1,\ldots,p_N\}: F(P_i)\ne 0\}.
\]
Switching to the complement subset, knowing the support is directly  equivalent to knowing the rational points on the variety defined by $F$,  \emph{i.e.}
\[\V(F)(\Fq)=\{P\in \PP^N(\Fq)~:~F(P)=0\}.\] 
Thus, to determine the minimal codewords of $\prm_q(2,N)$ we are looking for nonzero quadratic forms $F\in \Fq[X_0,\ldots,X_N]$ such that there is no other quadratic form $F'\in \Fq[X_0,\ldots,X_N]$ satisfying $\V(F)(\Fq)\subsetneq \V(F')(\Fq)$. That is, we look for quadrics whose sets of rational points are maximal with respect to the inclusion.

\subsection{Projective classification of quadrics in $\PP^N(\Fq)$}\label{subsec:Classes}

A \emph{quadric} in $\PP^N$ defined over $\Fq$ is the zero locus in $\PP^N$ of a quadratic form in $\Fq[X_0,\ldots,X_N]$. A quadric in $\PP^2$ is called a \emph{conic}.

  \medskip
  
\noindent \textbf{Caution.} In the present article, most of the
varieties we deal with are defined over $\Fq$. Therefore, whenever we
refer to a \emph{variety, quadric, conic} of any linear space, unless
otherwise specified they will be defined over $\Fq$.

\medskip

As a simple consequence of the Hilbert Nullstellensatz,
the quadratic forms in $\Fq[X_0,\ldots,X_N]$ up to multiplication by a nonzero scalar are in one--to--one correspondence with the quadrics in $\PP^N$.
We should however consider separately the case of
quadratic forms that are squares of linear forms,
the corresponding ideal being non radical. In this case,
the corresponding variety is a hyperplane that we will
refer to \emph{double hyperplane} since it is a hyperplane
``with multiplicity two''.

From now on, whenever convenient, we will talk in the language of quadrics in $\PP^N$ wherever convenient. The following classification of quadrics is folklore, and we include a proof for the sake of completeness. 

\begin{prop}\label{prop:ClassQuadrics}
  The quadrics in $\PP^N$ can be classified as follows:
    \begin{enumerate}
    \item\label{item:double_hyperplane} double hyperplanes (defined by $L^2$ where $L$ is a linear form). 
    \item\label{item:pair_of_hyperplanes} unions of distinct
      hyperplanes (defined by
      $L_1 L_2$ where $L_1, L_2$ are linearly independent linear forms).
    \item\label{item:conjugate_hyperplanes} $\Fq$-irreducible quadrics
      which are $\mathbb{F}_{q^2}$-reducible (defined by
      $L_1L_2$ where $L_1, L_2$ are linear forms with coefficients in
      $\F_{q^2}$ and conjugated by the Frobenius).
\item\label{item:absolutely_irreducible} absolutely irreducible quadrics. \end{enumerate}
\end{prop}

\begin{proof}
	Let $\Q$ be a quadric in $\PP^N$ defined by $F\in {\Fq[X_0,\ldots,X_N]}_2$. 
	If $\Q$ is an $\Fqbar$--irreducible quadric, it is in class (\ref{item:absolutely_irreducible}). Thus, we may assume that $\Q$ is $\Fqbar$--reducible, which implies that $F=L_1 L_2$ for $L_1,L_2\in {\Fqbar[X_0,\ldots,X_N]}_1$. For any $G\in \Fqbar[X_0, \dots, X_N]$, we denote by $G^{[q]}$ the form obtained from $G$ by raising all the coefficients to the $q$--th power. Since $F$ has coefficients
    in $\Fq$, we obtain
    \[
    F=F^{[q]}=L_1^{[q]}\cdot L_2^{[q]}.
    \]
    Since $\Fqbar[X_0,\ldots,X_N]$ is a unique factorisation domain, either $L_1^{[q]}=\alpha L_1$ and $L_2^{[q]}=(1/\alpha)L_2$ or $L_1^{[q]}=\beta L_2$ and  $L_2^{[q]}=(1/\beta)L_1$ for $\alpha,\beta \in \Fqbartimes$ respectively. In the first case, $L_1,L_2$ are $\Fqbar$-multiples in $\Fq[X_0,\ldots,X_N]$, leading to classes (\ref{item:double_hyperplane}) or (\ref{item:pair_of_hyperplanes}) depending on whether $L_1$ divides $L_2$ or not. In the second case, $L_i^{[q^2]}=L_i$ for $i=1,2$, \emph{i.e.} $L_1,L_2\in {\F_{q^2}[X_0,\ldots,X_N]}_1 \setminus {\F_{q}[X_0,\ldots,X_N]}_1$, yielding class (\ref{item:conjugate_hyperplanes}).
\end{proof}

\subsection{Rank and classification}

\begin{defn}\label{def:rank}
  The \emph{rank of a quadratic form} $F$ is the least number of
  variables needed in its expression, among all the possible
  $\Fq$--linear changes of variables of $\Fq[X_0,\ldots,X_N]$. The
  \emph{rank of a quadric} $\Q \subset \PP^N$ is that of any of its defining
  polynomials in $\Fq[X_0,\ldots,X_N]_2$.
\end{defn}

Most of the quantities of interest: rank, number of rational points,
\emph{etc.}  are invariant by linear changes of variables, \emph{i.e.}
under the action of the projective general linear group
$\mathbf{PGL}_N(\Fq)$ into \emph{canonical forms}.  The classes of
absolutely irreducible quadrics under this action have canonical
representatives described in the statement below.

\begin{thm}[{Hirschfeld and Thas~\cite{hirschfeld1991general}, Section 1.1}]
\label{thm:IrrClassQuadrics}
Let $\Q$ be an absolutely irreducible quadric in $\PP^N$. Then it has one of the following canonical
forms with respect to the action of $\mathbf{PGL}_N(\Fq)$.
    \begin{enumerate}
        \item $\P_{r}=\V(X_0^2+X_1 X_2+X_3 X_4\cdots+ X_{r-2}X_{r-1})$ for $r\geq 3$ odd
    \item For $r\geq 4$ even,  
        \begin{enumerate}
            \item $\H_r=\V(X_0X_1+X_2 X_3+\cdots+X_{r-2} X_{r-1})$
            \item $\E_r=\V(f(X_0,X_1)+X_2 X_3+\cdots+X_{r-2} X_{r-1})$
              where $f$ is an irreducible binary quadratic form, in particular, $f$ can be written as $X_0^2+dX_1^2$ where $d$ is a non--square when $p\ne 2$ and as $X_0^2+X_0 X_1+ d X_1^2$ where ${\rm Tr}_{\Fq/\F_{p}}(d)=1$ when $p=2$. 
        \end{enumerate}
      \end{enumerate}
      The quadrics $\P_r$, $\H_r$ and $\E_r$ are called the \emph{parabolic}, \emph{hyperbolic}, and \emph{elliptic}, respectively. 
      Moreover, the integer $r$ is the rank of $\Q$.
\end{thm}

Since the properties of quadrics we will use depend only on their
class modulo the action of $\mathbf{PGL}_N(\Fq)$, by abuse of notation
we allow ourselves to denote by $\P_r$, $\E_r$ and $\H_r$ any quadric
equivalent to $\P_r$, $\E_r$ and $\H_r$, respectively.

\begin{rmk}\label{rmk:CanonicalFormsRed}
  \rm{To complete the previous statement, forms $X_0^2$, $X_0X_1$, and $f(X_0, X_1)$, where $f \in \Fq[X_0,X_1]$ is irreducible
    yield respective canonical forms for the classes (\ref{item:double_hyperplane}), (\ref{item:pair_of_hyperplanes}), and (\ref{item:conjugate_hyperplanes}) in Proposition~\ref{prop:ClassQuadrics}.}
\end{rmk}

Finally, we can characterize quadrics according to their rank, as follows. 
\begin{prop}\label{prop:Rk}
	Let $\Q$ be a quadric in $\PP^N$ among the classes of Proposition~\ref{prop:ClassQuadrics}. 
	\begin{enumerate}
	\item \label{item:rkDoubleHyperplane}
	If $\Q$ is from class (\ref{item:double_hyperplane}), then $\rk \Q=1$;
	\item \label{item:rkHyperplanes}If $\Q$ is from class (\ref{item:pair_of_hyperplanes}) or
      (\ref{item:conjugate_hyperplanes}), then $\rk \Q=2$;
	\item \label{item:rkAbsIrr}If $\Q$ is from class (\ref{item:absolutely_irreducible}), then for $r\geq 3$ odd, $\rk \P_r=r$ and for $r\geq 4$ even, $\rk \E_r=\rk \H_r=r$.   
	\end{enumerate}
\end{prop}
\begin{proof}
Follows from Remark~\ref{rmk:CanonicalFormsRed}, Theorem~\ref{thm:IrrClassQuadrics}. 
\end{proof}

\begin{rmk}\label{rmk:charnot2matrix}
  \rm{Suppose $p\ne 2$. Then, for any quadratic form
    \[F=\sum_{0\leq i\leq j\leq N} a_{ij} X_i X_j,\]
    we may write
    $F=X^\top A X$ where $X=(X_i)_{0\leq i\leq N}$ and
    $A=(A_{i,j})_{1\leq i,j \leq N}$ is a symmetric matrix
    defined by
	\[A_{i,j}=A_{j,i}=\begin{cases}
			a_{ij}/2 & \text{if } i\ne j\\
a_{ij} & \text{if }i=j.
		\end{cases}
      \]
    Furthermore, $A$ can be reduced to a diagonal matrix, in particular, \[F=(PX)^\top D (PX)\]
    for some $P\in \mathbf{PGL}_N(\Fq)$ where $D=diag(1,\ldots,1,\lambda,0,\ldots,0)$ for $\lambda=1$ or a non--square in $\Fq$. When considering
    quadrics, \emph{i.e.} quadratic forms up to multiplication by a nonzero scalar, one gets the following possible canonical forms: either $X_0^2+X_1^2+\cdots+X_{r-1}^2$ if $r$ is even or $X_0^2+X_1^2+\cdots+X_{r-2}^2+\lambda X_{r-1}^2$ where $\lambda$ is either $1$ or a non-square in $\Fq$ if $r$ is odd. Finally, observe that the rank of the quadratic form is nothing but that of $A$.

    Note that in even characteristic, these forms are $\Fqbar$--reducible and hence cannot describe absolutely irreducible quadrics. Dickson~\cite{dickson1899determination} addressed this question for even characteristic resulting in the canonical forms mentioned in Theorem~\ref{thm:IrrClassQuadrics}. To the best of our knowledge, Hirschfeld~\cite{hirschfeld1998projective} rewrote the canonical forms in odd characteristic to unify the forms for both characteristics.}
\end{rmk}

\subsection{Singular locus and cone structure}
For any quadric $\Q$, we denote its singular locus by $\sing \Q$.
The rank of a quadric can be interpreted in terms of the dimension of its singular locus, as follows.

\begin{prop}\label{prop:rankSing}
  Let $F$ be a quadratic form in $\Fq[X_0,\ldots,X_N]$ and let $\Q=\V(F)$ be the corresponding quadric in $\PP^N$. Then, the
  singular locus of $\Q$ is a linear subspace and 
     \[
    \rk \Q=N-\dim \sing(\Q),
  \]
  with the convention that an empty singular locus has dimension $-1$.
\end{prop}
\begin{proof}
  After a suitable linear change of coordinates, $F$ is one of the
  canonical forms of Theorem~\ref{thm:IrrClassQuadrics} and hence is
  expressed in the variables $X_0, \dots, X_{r-1}$, where
  $r = \rk \Q$. A calculation of partial derivatives shows that the
  singular locus is $\V (X_0, \dots, X_{r-1})$.
\end{proof}

\begin{prop}[{\cite[Section~1.1]{hirschfeld1991general}}]\label{prop:SingCone}
Any singular quadric $\Q$ in $\PP^N$ is a cone on its singular locus $\sing(\Q)$ and a smooth quadric in a linear subspace $\Pi\subset \PP^N$ such that $\Pi \cap \sing (\Q) = \emptyset$ and
  $\dim \Pi + \dim \sing (\Q) = N-1$. By \emph{cone} we mean that $\Q$ is the union of all the possible lines in
  $\PP^N$ joining a point of $\Q \cap \Pi$ and a point of $\sing \Q$.
\end{prop}

Note that in the statement above, $\Pi$ is not unique but the
equivalence class of $\Q \cap \Pi$ with respect to the action of the
projective linear group does not depend on the choice of $\Pi$. In the reverse
direction, the following statement expresses defining polynomials from the knowledge
of the singular locus.

 \begin{prop}\label{prop:ToBaseQuadric}
	Let $F$ be a quadratic form in $\Fq[X_0,\ldots,X_N]$ and let $\Q$ be the quadric $\V(F)$ in $\PP^N$ of rank $r$.  If $\sing(\Q)=\V(X_0,\ldots,X_{r-1})$, then $F$ is a quadratic form in $\Fq[X_0,\ldots,X_{r-1}]$.  
\end{prop}

\begin{proof}
	Write $F=\sum_{0\le i,j\le N} a_{ij} X_i X_j$. Since $F(u)=0$ for $u\in \sing(\Q)$,  we must have $F(0,\ldots,0,X_r,\ldots,X_N)=0$, \emph{i.e.} 
	\begin{equation}\label{eq:Fform}
		F=\sum_{0\le i,j<r} a_{ij} X_i X_j+\sum_{\substack{0\le i<r\\ N\ge j\ge r}} a_{ij} X_i X_j
	\end{equation}	
	Let $i_0<r$. For $t\ge r$, define $e_t\eqdef (0:\dots:0:1:0:\dots:0)\in \sing(\Q)$ where $1$ is at the $(t+1)$--th position. Then, 
	\[
		\frac{\partial F}{\partial X_{i^{}_0}}(e_t)=0 \quad \text{while} \quad \frac{\partial F}{\partial X_{i^{}_0}}=\sum_{0\le j\le N} a_{i^{}_0j} X_j.
	\]
	Thus, $a_{i^{}_0 t}=0$. In particular, $a_{ij}=0$ for all $i<r$ and $j\ge r$. Hence, by \eqref{eq:Fform}, we obtain 
	\[
		F=\sum_{0\le i,j<r} a_{ij} X_i X_j.
	\]
\end{proof}

\subsection{Projective index}
   The \emph{projective index} of a quadric $\Q$ in $\PP^N$ is the largest dimension of a projective linear subspace defined over $\Fq$ and contained in  $\Q$. It is denoted by $g(\Q)$.

It is clear that for any quadric of class (\ref{item:double_hyperplane}) or (\ref{item:pair_of_hyperplanes}) in Proposition~\ref{prop:ClassQuadrics} \emph{i.e.} double hyperplanes and unions of distinct hyperplanes, the projective index is $N-1$. For unions of conjugate hyperplanes (class (\ref{item:conjugate_hyperplanes})), the projective index is $N-2$. The following proposition describes the projective index of quadrics of higher rank.

\begin{prop}[{\cite[Cor. 1]{hirschfeld1998projective}}]\label{prop:ProjIndex}
For $1\le s\le \lfloor N/2\rfloor$, the projective indices of $\P_{2s+1}$, $\H_{2s}$ and $\E_{2s}$ in $\PP^N$ are $N-s-1$, $N-s$ and $N-s-1$, respectively.
\end{prop}

\subsection{Number of rational points}

The number of rational points on quadrics of classes (\ref{item:double_hyperplane}),  (\ref{item:pair_of_hyperplanes}), and (\ref{item:conjugate_hyperplanes}) in Proposition~\ref{prop:ClassQuadrics} can be computed as follows. For class (\ref{item:double_hyperplane}), it equals the number of rational points on a hyperplane. For class (\ref{item:pair_of_hyperplanes}), it equals  the number of rational points on a union of two hyperplanes. The latter number is twice the number of rational points of a hyperplane, minus that of their intersection which is a linear space of codimension $2$. For class (\ref{item:conjugate_hyperplanes}), it equals the number of rational points on a codimension $2$ linear subspace, as the rational points are those on the intersection of the two $\F_{q^2}$-hyperplanes defining the quadric. These numbers are $p_{N-1}$, $2q^{N-1}+p_{N-2}$ and $p_{N-2}$, respectively.

By Primrose~\cite{primrose1951quadrics}, 
we also know the number of rational points on absolutely irreducible quadrics, which is as follows.

\begin{prop}\label{prop:NoOfPoints}
For $1\le s\le \lfloor N/2\rfloor$, the number of rational points on absolutely irreducible quadrics in $\PP^N$ are
    \begin{enumerate}
        \item $\displaystyle\sharp\P_{2s+1}(\Fq)= p_{N-1}$;
        \item $\displaystyle\sharp\H_{2s}(\Fq)= p_{N-1}+q^{N-s}$;
        \item $\displaystyle\sharp\E_{2s}(\Fq)= p_{N-1}-q^{N-s}$.
        \end{enumerate}
\end{prop}

The classical Serre bound~\cite{serre1989lettrea}, which gives the
maximum number of rational points on a hypersurface of degree
$d \leq q$, as well as the hypersurfaces that attain it, can be
recovered in the case of quadrics using
Proposition~\ref{prop:NoOfPoints}, and the previous discussion. We
recall it below.

\begin{thm}[Serre bound \cite{serre1989lettrea}, degree 2 case]\label{thm:Serre}
  The number of rational points on a quadric in $\PP^N$ is bounded by $2q^{N-1}+p_{N-2}$. The bound is attained precisely for the quadrics defined by union of distinct hyperplanes.
\end{thm}

\begin{cor}\label{cor:linear}
If a quadric $\Q$ contains all the rational points of a linear subspace $\Pi$, it contains the linear subspace.
\end{cor}
\begin{proof}
Assume $\Pi\not\subset \Q$, then  Theorem~\ref{thm:Serre} applied to $\Q\cap \Pi$ shows that $\sharp \Q\cap \Pi(\Fq)<\sharp \Pi(\Fq)$, a contradiction. 
\end{proof}

\subsection{Bilinear and quadratic forms}
    Let $V$ be a vector space over a field $K$. A \emph{bilinear form $B$ on $V$} is a bilinear map from $V\times V$ to $K$ and its \emph{radical} is 
    \[
    \rad B\eqdef \{v\in V: B(v,w)=0 \text{ for all }w\in V\}.
    \]

\begin{defn}
  The \emph{polarisation} of a quadratic form $F\in \Fq[X_0,\ldots,X_N]_2$ is the bilinear form:
    \[
    B_F(u,v)\eqdef \map{\FqbarN \times \FqbarN}{\Fqbar}{(u,v)}{F(u+v)-F(u)-F(v).}
    \]
\end{defn}

The polarisation $B_F$ of $F$ can be interpreted as the differential of $F$. Namely, for any $u \in \overline{\F}_q^{N+1}$, $B_F(u, \cdot)$ is the unique linear map satisfying
  \[\forall v \in \overline{\F}_q^{N+1},\quad F(u+v) = F(u) + B_F(u,v) + F(v)\]
  where the rightmost term is quadratic in $v$.

\begin{defn}    
Let $F$ be a quadratic form in $\Fq[X_0,\ldots,X_N]$ which polarizes to $B_F$. The \emph{radical of $F$} is the $\Fqbar$--space 
    \[
    \rad F\eqdef \{v\in \rad B_F: F(v)=0\}.
    \]
\end{defn}

Note that in odd characteristic we always have $\rad F = \rad B_F$ but this equality may fail in
characteristic $2$.

\begin{lem}\label{lem:SingRad}
      Let $F$ be a quadratic form in $\Fq[X_0,\ldots,X_N]$ and let $\Q$ be the quadric $\V(F)$ in $\PP^N$. Then the singular locus $\sing \Q$ satisfies
      \[\sing (\Q) = \PP (\rad (F)).\]
\end{lem}

    \begin{proof}
      Let us first consider the case where $F = X_i X_j$ for some
      possibly equal $i,j\in \{0,1,\ldots,N\}$.
     Then, 
\[\forall u,v \in \overline{\F}_q^{N+1},\quad     B_F(u,v)=u_i v_j+u_j v_i =\frac{\partial F}{\partial X_i}(u) v_i+ \frac{\partial F}{\partial X_j}(u)  v_j.\]
Hence, for any $F\in {\Fq[X_0,\ldots,X_N]}_2$, 
\[B_F(u,v)=  \left\langle \left(\frac{\partial F}{\partial X_0}(u),\ldots,\frac{\partial F}{\partial X_N}(u)\right),(v_0,\ldots,v_N)\right\rangle,\]
where $\langle \cdot, \cdot \rangle$ denotes the canonical Euclidean product on $\overline{\F}_q^{N+1}$.

Now, $B_F(u,v)=0$ for every $v\in \FqbarN$ if and only if 
$(\frac{\partial F}{\partial X_0}(u),\ldots,\frac{\partial F}{\partial X_N}(u))=0$. In particular, $u\in \rad F$ if and only if $F(u)=0$ and $(\frac{\partial F}{\partial X_0}(u),\ldots,\frac{\partial F}{\partial X_N}(u))=0$, \emph{i.e.} if and only if $\PP(u)\in \sing(\Q)$.

The general case is deduced by linearity since the map $F\mapsto B_F$ is linear.
    \end{proof}

\begin{prop}\label{prop:rankRad}
 Let $F$ be a quadratic form in $\Fq[X_0,\ldots,X_N]$ and let $\Q=\V(F)$ be the corresponding quadric in $\PP^N$. Then the rank of $\Q$ (or $F$),
  introduced in Definition~\ref{def:rank}, satisfies
     \[
    \rk \Q=N-\dim \sing(\Q)=N-\dim_{\Fqbar} \rad F+1.  \]
\end{prop}
\begin{proof}
  Follows from Proposition~\ref{prop:rankSing} and Lemma~\ref{lem:SingRad}.
\end{proof}

The following proposition is given as an exercise in \cite[Ex.~22.1]{Harris1992}. We prove
it for the sake of completeness.
 
 \begin{prop}\label{prop:HyperplaneCutRk}
   Let $F$ be a quadratic form in $\Fq[X_0,\ldots,X_N]$. Denote by $\Q$ the quadric $\V(F)$ in $\PP^N$. For any hyperplane $\L$, 
   \[ \rk(\Q) - 2\leq \rk (\Q \cap \L) \leq \rk (\Q).
   \]  
 \end{prop}
\begin{proof}
     Without loss of generality, we assume that $\L=\V(X_N)$. Denote by $F_\L$ the quadratic form $F(X_0,\ldots,X_{N-1},0)$ in $\Fq[X_0,\ldots,X_{N-1}]$ defining the quadric $\Q\cap \L$ in  $\L\simeq \PP^{N-1}$. Define the $\Fqbar$-linear map 
     \[
          \phi: \map{\rad F}{\Fqbar}{(v_0,\ldots,v_N)}{v_N.}
        \]
        On one hand, $\ker \phi$ is either equal to $\rad F$ or to a hyperplane of $\rad F$.
        On the other hand,
     any element of $\ker \phi$ has form $(v_0, \dots, v_{N-1},0)$ where
     $(v_0, \dots, v_{N-1}) \in \rad {F_\L}$. The two previous observations entail 
    \begin{equation}\label{eq:radKer1}
    	\dim_{\Fqbar} \rad F-1 \leq \dim_{\Fqbar} \ker \phi\leq \dim_{\Fqbar} \rad F_{\L}.
    \end{equation} 
    Next, let $(v_0,\ldots,v_{N-1})\in \rad F_{\L}$. Therefore, 
    \[F(v_0,\ldots,v_{N-1},0)=0\quad \text{and} \quad B_F((v_0,\ldots,v_{N-1},0),(w_0,\ldots,w_{N-1},0))=0
    \]
    for all $(w_0,\ldots,w_{N-1})\in \overline{\F}^{N}_q$. It implies that $(v_0,\ldots,v_{N-1},0)\in \rad F$ if we have $B_F((v_0,\ldots,v_{N-1},0),(0,\ldots,0,1))=0$. 
    So, consider the linear form
    \[
          \psi: \map{\rad F_{\L}}{\Fqbar}{(v_0,\ldots,v_{n-1})}{B((v_0,\ldots,v_{n-1},0),(0,\ldots,0,1)).}
     \]
We deduce that
    \begin{equation}\label{eq:radKer2}
    	\dim \rad F_{\L}-1\leq \dim \ker \psi\leq \dim \rad F.
    \end{equation}
    Combining \eqref{eq:radKer1} and \eqref{eq:radKer2}, we obtain
    \[
    N+1-(\dim \rad F_{\L}+1)\leq N+1-\dim \rad F\leq N+1-(\dim \rad F_{\L}-1)
    \]
    which implies that \[
    \rk F_{\L}\leq \rk F\leq \rk F_{\L}+2,
    \]
    as required. 
    \end{proof}

    \begin{prop}\label{prop:SameRk}
      Let $F$ be a quadratic form in ${\Fq[X_0,\ldots,X_N]}$ and let
      $\Q=\V(F)$ in $\PP^N$ with a non empty singular locus. Consider a hyperplane
      $\L\subset \PP^N$ which does not contain $\sing (\Q)$. Then
      $\rk (\Q \cap \L) = \rk (\Q)$.
    \end{prop}
    \begin{proof}
      After a possible linear change of variables, it suffices to assume that the point
      $(0:\cdots:0:1) \in \sing(\Q)$ and $\L = \V (X_N)$. Denote
      by $F_\L \eqdef F(X_0, \dots, X_{N-1},0)$.  By
      Lemma~\ref{lem:SingRad}, it follows that $(0,\ldots,0,1)\in \rad F$. In particular, for every $(v_0,\ldots,v_{N-1})\in \rad F_\L$, we obtain $(v_0,\ldots,v_{N-1},0)\in \rad F$. Moreover, since $(0,\ldots,0,1)\in \rad F$, we have $\dim_{\Fqbar} \rad F\ge \dim_{\Fqbar} \rad F_\L+1$. But by Propositions~\ref{prop:HyperplaneCutRk} and \ref{prop:rankRad}, $\dim_{\Fqbar} \rad F\le \dim_{\Fqbar} \rad F_\L+1$, which proves the equality.
    \end{proof}

\subsection{Tangent spaces}
In the sequel
we make use of \emph{tangent spaces} which will always refer
to as \emph{projective tangent spaces} as defined in \cite[Lect.~14]{Harris1992}. The notion can be
defined for any projective variety but we will only use it for quadric
hypersurfaces.  Given a hypersurface $\mathcal S \subset \PP^N$
defined as $\mathcal S = \V (F)$ for some homogeneous polynomial
$F\in \Fq[X_0, \dots, X_N]$ and a point $P \in \mathcal S$, we define
\[
\T_P \mathcal S \eqdef \left\{(x_0:\cdots:x_N) \in \PP^N ~:~ \sum_{i=0}^N x_i \frac{\partial F}{\partial X_i}(P) = 0\right\}.
\]
From Euler theorem on homogeneous polynomials, $P \in \T_P \mathcal S$
and it is well--known that $\T_P \mathcal S$ is a hyperplane of
$\PP^N$ if and only if $\mathcal S$ is smooth at $P$. Otherwise, when
$P$ is singular then $\T_P \mathcal S = \PP^N$.

With this definition at hand, a line $\ell$ through $P$ is
\emph{tangent} to $\mathcal S$ at $P$ if
$\ell \subset \T_P \mathcal S$. In particular, if $P$ is a singular
point of $\mathcal S$, then any line through $P$ will be said to be
tangent to $\mathcal S$ at $P$. The following statement is a direct consequence of Bézout's Theorem.

\begin{lem}\label{lem:tangent_line}
  Let $\Q \subset \PP^N$ be a quadric hypersurface, $\ell \subset \PP^N$
  be a line with $\ell \not\subset \Q$ and $P \in \ell \cap \Q$. Then $\ell$ is tangent to $\Q$ at $P$
  if and only if $\ell \cap \Q = \{P\}$.
\end{lem}

Finally, we will make use of the following very classical
statement which can be proven in whole generality and for which a
proof for the specific case of quadrics (which is the purpose of the
present article) is given in \cite[Thm.~1.11(ii)]{hirschfeld1991general}.

\begin{lem}\label{lem:inclusion_tangent}
  Let $\mathcal S$ be a hypersurface of $\PP^N$ and
  $P \in \mathcal S$.  Let $\Pi$ be a linear subspace of $\PP^N$ such
  that $P \in \Pi \subset \mathcal S$.
  Then, $\Pi \subset \T_P \mathcal S$.
\end{lem}

\section{Quadrics with maximal rational points}\label{sec:max_quadrics}

As noted in Subsections~\ref{subsec:codes}, to find the minimal codewords of $\prm_q(2,N)$, we need to know the quadrics in $\PP^N$ with maximal set of rational points with respect to the inclusion. Thus, we ask ourselves the following question: if $\Q$ and $\Q'$ are two quadrics in $\PP^N$ such that $\Q(\Fq)\subset \Q'(\Fq)$, what can we say about the quadrics $\Q$ and $\Q'$? From Proposition~\ref{prop:ClassQuadrics}, rational points on quadrics from class (\ref{item:double_hyperplane}) are clearly contained in some quadrics from class (\ref{item:pair_of_hyperplanes}). Moreover, since the rational points on any quadric from class (\ref{item:conjugate_hyperplanes}) are the rational points on codimension $2$ linear subspace, it is contained in some hyperplane and hence in  quadrics of class (\ref{item:double_hyperplane}) or (\ref{item:pair_of_hyperplanes}). Thus, we want to compare the quadrics which are either absolutely irreducible or a union of distinct hyperplanes. We begin by identifying some properties of quadrics satisfying $\Q(\Fq)\subset \Q'(\Fq)$. 
We first present some technical tools that will be useful in the sequel.

\subsection{Some technical lemmas}

\begin{lem}\label{lem:rationalpoints}
  Let $\Q$ be a quadric in $\PP^N$ and let $P$ be a rational
  point on $\Q$. If $\Q$ is smooth at $P$, then any $\Fq$--line $\ell$ through $P$
  that is non-tangent to $\Q$ at $P$ intersects $\Q$ at exactly two
  rational points, of which one is $P$. Equivalently, if $(\ell\cap \Q)(\Fq)=\{P\}$ for an $\Fq$-line $\ell$, then $\ell\cap \Q=\{P\}$, \emph{i.e.} $\ell$ is tangent to $\Q$ at $P$. 
\end{lem}

\begin{lem}\label{lem:Qhyperplane}
  Let $\Q$ be a quadric in $\PP^N$. Then, $\Q(\Fq)$ is contained in
  a hyperplane if and only if $\Q$ is either a double hyperplane (class (\ref{item:double_hyperplane}) of Proposition~\ref{prop:ClassQuadrics}) or an
  irreducible non-absolutely irreducible quadric (class (\ref{item:conjugate_hyperplanes})).
\end{lem}

\begin{proof}
  The ``if'' part is straightforward.
  
  Note that if $\Q$ is reducible (class
  (\ref{item:pair_of_hyperplanes}) of
  Proposition~\ref{prop:ClassQuadrics}), then it is the union of two
  distinct rational hyperplanes and cannot have all its rational
  points inside a hyperplane. Suppose now that $\Q$ is not a double
  hyperplane, is absolutely irreducible and $\Q(\Fq)$ is contained in
  some hyperplane $\L$ of $\PP^N$.  Let $P \in \Q(\Fq)$, then any line
  $\ell$ such that $P\in \ell$ and $\ell \not \subset \L$ intersects
  $\L$ only at $P$. In particular, $\Q(\Fq)\cap \ell=\{P\}$. Thus by Lemma~\ref{lem:rationalpoints}, it intersects
  $\Q$ only at $P$ and any such line is tangent to $\Q$ at $P$.
  Since the span of these lines is $\PP^N$ itself, we deduce that
  $\T_P \Q$ is not a hyperplane and hence that $P$ is a singular point
  of $\Q$. Thus,
  \begin{equation*}\Q (\Fq) \subset \sing (\Q).
  \end{equation*}
  Since $\Q$ is absolutely irreducible and is a cone over a smooth quadric in $\PP^{r-1}$ for $r\ge 3$ by Proposition~\ref{prop:SingCone}, $\Q$ has a smooth rational point, a contradiction. 
\end{proof}

\begin{prop}\label{prop:SingPts}
  Let $\Q$ and $\Q'$ be two quadrics in $\PP^N$ such that
  $\Q(\Fq) \subset \Q'(\Fq)$ and $\Q$ is absolutely irreducible. Then
  $\sing(\Q) \subset \sing(\Q')$.
\end{prop}

\begin{proof}
  Since singular loci are linear spaces, if
  $\sing(\Q)(\Fq) \subset \sing(\Q')(\Fq)$, then
  $\sing(\Q) \subset \sing(\Q')$. Therefore, we only have to prove the
  statement on rational points.

  Let $P \in \sing(\Q)(\Fq)$ and suppose that
  $P \notin \sing(\Q')$. Then $\Q$ is smooth at $P$ and the tangent space
  $\T_P \Q'$ is a hyperplane.  Since
  $\Q$ is absolutely irreducible, from
  Lemma~\ref{lem:Qhyperplane} the rational points of $\Q$ cannot be
  all contained in $\T_P \Q'$. Hence, let $P_2 \in \Q (\Fq)$ and
  $P_2 \notin \T_P \Q'$.
  The line $\ell$ joining $P$ to $P_2$ is not
  contained into $\T_P \Q'$. However, $\ell \cap \Q$ contains $P, P_2$ and since $P$ is a singular point of $\Q$, the line $\ell$ is tangent to $\Q$ at $P$. Then, it follows from Lemma~\ref{lem:tangent_line} that $\ell \subset \Q$.
  Next, since $\Q (\Fq) \subset \Q' (\F_q)$, then
  $\ell (\Fq)\subset \Q'$ which further entails 
  $\ell \subset \Q'$ by Corollary~\ref{cor:linear}.  However, from
  Lemma~\ref{lem:inclusion_tangent}, $P\in \ell \subset \Q'$ implies
  that $\ell\subset \T_P \Q'$. A contradiction.
\end{proof}

\begin{lem}\label{lem:NonMaximalCases}
Let $\Q$ be a quadric in $\PP^N$. In the following cases, $\Q$ does not have a maximal set of rational points among all quadrics in $\PP^N$. 
\begin{enumerate}
\item\label{item:q_more_than_3} $q\leq 3$ and $\Q$ has rank $3$. 
\item\label{item:q=2} $q=2$ and $\Q$ is an elliptic quadric of rank $4$.
\end{enumerate}
\end{lem}
\begin{proof}
    (\ref{item:q_more_than_3}) Without loss of generality, suppose $\Q=\V(X_0^2+X_1 X_2)$. Then,
    \[\Q(\Fq)\subset \V(X_0(X_1+X_2))(\Fq) \quad \text{for } q\leq 3.\]
    Indeed, let $P = (x_0:\cdots:x_N) \in \Q(\Fq)$, if $x_0 = 0$ then
    obviously $P \in \V(X_0(X_1+X_2))$.
    Otherwise, one can normalise to $x_0=1$ and hence $x_1 = -x_2^{-1}$, which in $\Fq$ with $q \in \{2,3\}$ entails $x_1 +x_2 = 0$. Hence, any $P \in \Q(\Fq)$ is in $\V(X_0(X_1+X_2))$.

    (\ref{item:q=2}) Without loss
    of generality, suppose $\Q=\V(X_0^2+X_1^2+X_0 X_1+X_2 X_3)$. Then
    \[\Q(\F_2)\subset \V(X_0(X_2+X_3))(\F_2).\]
\end{proof}

\subsection{Main statements}

Before stating the main theorems, we recall the Segre embedding
\begin{equation}\label{eq:mapTheta}
          \theta: \map{\PP^1\times \PP^1}{\PP^3}{((a_0:a_1), (b_0:b_1))}{(a_0 b_0: a_0 b_1: a_1 b_0: a_1 b_1).}
        \end{equation}
     This map is injective, and its image $\im\theta$ is precisely the hyperbolic quadric in  $\PP^3$ defined by $X_0 X_3-X_1X_2$.

\begin{thm}\label{thm:irredandred}
    Let $\Q$ be an absolutely irreducible quadric of $\PP^N$ and let $\Q'$ be a union of distinct rational hyperplanes in $\PP^N$. Then the rational points of $\Q$ and $\Q'$ are not contained in one another, except for the cases of Lemma~\ref{lem:NonMaximalCases}.
\end{thm}

\begin{proof}
  Assume that $\Q'=\L_1\cup \L_2$ for distinct rational hyperplanes
  $\L_1$ and $\L_2$ in $\PP^N$. Pick a new hyperplane, say $\L_3$, from
  the pencil of $\L_1$ and $\L_2$, that is to say a hyperplane distinct
  from $\L_1, \L_2$ which contains $\L_1 \cap \L_2$. Such a hyperplane
  exists since, without loss of generality, one can assume that
  $\L_1 = \V (X_1)$, $\L_2 = \V(X_2)$ and one can take
  $\L_3 = \V (X_1+X_2)$. Note that
  $\L_3 \cap (\L_1 \cup \L_2) = \L_1 \cap \L_2 = \L_1 \cap \L_3$.

  Then, since $\Q(\Fq) \subset (\L_1 \cup \L_2)(\Fq)$,
  \[\Q\cap \L_3(\Fq)\subset (\L_1\cup \L_2)\cap \L_3(\Fq)=(\L_1 \cap \L_2)(\Fq) = (\L_1 \cap \L_3)(\Fq).\]
  Regarding $Q \cap \L_3$ as a quadric of $\L_3 \simeq \PP^{N-1}$, we
  see that $(\Q\cap \L_3)(\F_q)$ is contained in $\L_1 \cap \L_3$ which is
  a hyperplane of $\L_3 \simeq \PP^{N-1}$.  By
  Lemma~\ref{lem:Qhyperplane}, $\Q\cap \L_3$ is either a double
  hyperplane or an irreducible non-absolutely irreducible quadric of
  $\L_3$. This implies, by
  Proposition~\ref{prop:Rk}, that $\rk (\Q\cap \L_3)$ equals either $1$ or $2$, while $\Q$
  is absolutely irreducible and hence has rank $\geq 3$. This,
  together with Proposition~\ref{prop:HyperplaneCutRk} implies that $\rk \Q$ equals either $3$
  or $4$. We proceed by considering the two aforementioned cases.

  \medskip

{\it Case 1.} $\rk \Q=3$.

Suppose $N\geq 3$. Then by Lemma~\ref{lem:SingRad},
$\dim \sing(\Q)\geq 0$.  Moreover, by Proposition~\ref{prop:SingPts},
$\sing\Q\subset \sing(\Q')$.
Now, we can pick a hyperplane, call it
$\L$, which does not contain $\sing(\Q')$. From
Proposition~\ref{prop:SameRk}, $\rk (\Q\cap \L)=\rk \Q$ and
$\rk (\Q'\cap \L)=\rk \Q'$ in $\PP^{N}\cap \L\simeq \PP^{N-1}$. This
hyperplane cut out process can be iterated until we reach $\PP^2$. Thus, we are
reduced to a problem in $\PP^2$ with $\Q$ parabolic of rank $3$ and $\Q'$ a union of
two distinct lines in $\PP^2$. Such a plane conic has $q+1$ rational points such that no $3$
of them are collinear. Hence, for $q\geq 4$,
$\Q(\Fq)\not\subset (\L_1\cup \L_2)(\Fq)$.  For $q = 2$ and $3$, we
land in the case of
Lemma~\ref{lem:NonMaximalCases}(\ref{item:q_more_than_3}).

\medskip

{\it Case 2.} $\rk \Q=4$.

Suppose $N\geq 4$. Following the same steps as in Case~1, we can
reduce our problem to solving it in $\PP^3$ with $\Q$ either elliptic
or hyperbolic both of rank $4$ and $\Q'=\L_1\cup \L_2$. First assume that
$\Q$ is elliptic of rank
$4$.  From Proposition~\ref{prop:ProjIndex}, its projective index is
$0$, hence $\Q$ contains no
rational line and so do the plane conics $\Q \cap \L_1$ and $\Q \cap
\L_2$ which therefore can only be either parabolic quadric or unions
of conjugate $\F_{q^2}$-lines. Consequently, by Proposition~\ref{prop:NoOfPoints},
each plane section $\Q \cap \L_1$ and $\Q \cap\L_2$ has at most
$q+1$ rational points, hence
\[
\Q(\Fq) \subset (\Q\cap \L_1) \cup (\Q\cap \L_2)\quad
\implies \quad \sharp\Q(\Fq) \leq 2(q+1).
\]
But $\sharp\Q(\Fq)=q^2+1$ by Proposition~\ref{prop:NoOfPoints} which gives a contradiction for $q\geq 3$. 

Now suppose $\Q$ is hyperbolic of rank $4$. By the Segre embedding~\eqref{eq:mapTheta}, we
can write $\Q$ as a $\PP^1$--bundle over $\PP^1$. In particular
it contains $\sharp \PP^1 (\Fq) = q+1 \geq 3$ pairwise skew rational lines, \emph{i.e.} no two of them are coplanar in $\PP^3$. By Corollary~\ref{cor:linear}, any such line is in $\L_1 \cup \L_2$ and hence
at least two of them are coplanar, a contradiction.
\end{proof}

Recall Theorem~\ref{thm:main}:

\main*

The unique exception in Theorem~\ref{thm:main} can be made explicit as
follows. Take $\Q=\V(X_0^2 + X_0X_1 + X_1^2 + X_2X_3)$ and
$\Q' = \V(X_0(X_0+X_3)+X_1(X_1+X_2))$. One can check that
$\Q (\F_2) \subset \Q' (\F_2)$.

Before proving the above statement, let us restate an interesting consequence stated in the introduction.

\characterisation*

\begin{proof}[Proof of Corollary~\ref{cor:characterisation}]
  Clearly, irreducible non absolutely irreducible quadrics can not be characterised by their set of rational points since this set is that of a linear subspace of codimension $2$ and for such a linear
  subspace, one can find more than one pair of conjugate $\F_{q^2}$-- hyperplanes containing it.

  Suppose now that two absolutely irreducible quadrics $\Q$ and $\Q'$
  have the same set of rational points.  In particular,
  $\Q (\Fq) \subset \Q'(\Fq)$. According to
  Theorem~\ref{thm:main}, they should be equal since the only possible exception of
  Theorem~\ref{thm:main} corresponds to quadrics that do not have the same number of rational points
  (elliptic and hyperbolic quadrics over $\F_2$).

  If one of the quadrics, say $\Q'$, is a union of distinct rational
  hyperplanes, then by the equality case of Serre's bound
  (Theorem~\ref{thm:Serre}), so is $\Q$ and we deduce that the two
  quadrics should be equal.

  Finally, if one of the quadrics is a double hyperplane, then, from
  Lemma~\ref{lem:Qhyperplane}, they should also be equal.
\end{proof}

\subsection{Proof of Theorem~\ref{thm:main}}
The proof consists in an induction on $N$. 
Let $F$ and $F'$ be the quadratic forms defining $\Q$ and $\Q'$ and let $r$ and $r'$ be the respective ranks of $\Q$ and $\Q'$. From Proposition~\ref{prop:SingCone}, $\Q$ is a cone on $\sing(\Q)$ of dimension $N-r$ and a smooth quadric in an $(r-1)$--dimensional linear subspace $\Pi$ skew to $\sing(\Q)$. Without loss of generality, assume that $\sing(\Q)=\V(X_0,\ldots,X_{r-1})$ and $\Pi=\V(X_r,\dots,X_N)$. Now, since $\Q(\Fq)\subset \Q'(\Fq)$, by Proposition~\ref{prop:SingPts}, $\sing(\Q)=\V(X_0,\ldots,X_{r-1})\subset \sing(\Q')$. Thus, using Proposition~\ref{prop:ToBaseQuadric}, both $F$ and $F'$ can be written in $\Fq[X_0,\ldots,X_{r-1}]$. Therefore, we can consider $\Q\cap \Pi$ and $\Q'\cap \Pi$ in the space $\Pi$ to study the same quadratic forms $F$ and $F'$. Since $\Q\cap \Pi$ is smooth in $\Pi$, 
we can reduce our problem to solving the case when $\Q$ is smooth. We first prove the induction step, which carries the main mechanism of the proof, and then verify the base cases.

\subsubsection{The case $N\ge 4$}\label{subsec:InductionStep} Assume that the theorem holds for $N-1$. Since $\Q$ is smooth, then $r=N+1\ge 5$.  Note that $r\ge r'$ by Proposition~\ref{prop:SingPts} and Proposition~\ref{prop:rankSing}. Moreover, since $\Q'$ is absolutely irreducible, $r'\ge 3$ because of Proposition~\ref{prop:Rk}(\ref{item:rkAbsIrr}). 
     
Without loss of generality, assume that $F$ is one of the forms in Theorem~\ref{thm:IrrClassQuadrics}, \emph{i.e.} $X_0^2+X_1 X_2+\cdots+X_{r-2} X_{r-1}$ if $r$ is odd, $X_0 X_1+\cdots+X_{r-2}X_{r-1}$ or $X_0^2+\alpha X_0 X_1+d X_1^2+X_2 X_3+\cdots+X_{r-2} X_{r-1}$ if $r$ is even. In the latter case, $\alpha=1$ with $\text{Tr}_{\Fq/\F_p}(d)=1$ if $p=2$ and $\alpha=0$ with $d$ is a non square if $p\ne 2$. Remark that these forms of $\Q$ agree with the assumption $\Pi=\V(X_0,\ldots,X_{r-1})$ in the beginning of the proof. 
      Consider the hyperplane sections of $\Q$ and $\Q'$ with respect to $\V(X_0)$. Since $\Q(\Fq)\subset \Q'(\Fq)$, we obtain $\Q(\Fq)\cap \V(X_0)\subset \Q'(\Fq)\cap \V(X_0)$. Because of the defined form, $\Q\cap \V(X_0)$ is never elliptic (hence avoiding the exceptional case of the theorem) and has rank greater than or equal to $4$.  Moreover, by Proposition~\ref{prop:HyperplaneCutRk}, $\rk (\Q'\cap \V(X_0))\geq 1$. Now by Lemma~\ref{lem:Qhyperplane} and Theorem~\ref{thm:irredandred}, $\rk (\Q'\cap \V(X_0))\geq 3$, \emph{i.e.} it is absolutely irreducible. Thus, we can use induction hypothesis in the space $\V(X_0)\simeq \PP^{N-1}$ and derive that 
\[F'(0,X_1,\ldots,X_{N-1},X_N)=\lambda F(0,X_1,\ldots,X_{N-1},X_N)\quad \text{for some } \lambda \in \F^{\times}_q.\]  
Write 
\[F=X_0 L+F(0,X_1,\ldots,X_{N-1},X_N) \text{ and }  F'=X_0 L'+ \lambda F(0,X_1,\ldots,X_{N-1},X_N)\]
for some linear forms $L,L'\in \Fq[X_0,\ldots,X_N]$. Then $F'-\lambda F=X_0(L'-\lambda L)$. Since $\Q(\Fq)\subset \Q'(\Fq)$, we obtain 
\[\Q(\Fq)\subset \V(F'-\lambda F)(\Fq)=\V(X_0(L'-\lambda L))(\Fq),\]
which is not possible by Theorem~\ref{thm:irredandred} as $r\ge 5$ unless $L'=\lambda L$ and hence
$F'=F$.

\subsubsection{The case $N=2$}

In $\PP^2$, all the absolutely irreducible conics
are parabolic of rank
$3$ by Theorem~\ref{thm:IrrClassQuadrics}. Therefore, by Proposition~\ref{prop:NoOfPoints}, $\sharp\Q(\Fq)=q+1$. Note
that an arbitrary conic is given by
\begin{equation}\label{eq:conic}
a_0 X_0^2 + a_1 X_1^2 + a_2 X_2^2 + a_3 X_0 X_1 + a_4 X_1 X_2 + a_5 X_0 X_2 = 0,
\end{equation}
where $a_i \in \F_q$ for $0 \leq i \leq 5$. In particular, the set of all plane
conics\footnote{In algebraic geometry, such sets are referred as
  \emph{linear systems} of curves.} is parameterized by a projective
space of dimension $5$.
It is well--known that $5$ points such that no $4$ of them are
collinear impose $5$ independent conditions
on~\eqref{eq:conic}. Therefore:

\medbreak

\noindent \textbf{If $\mathbf{q\geq 4}$}, $\sharp \Q(\Fq)=q+1\geq 5$ and no $3$ of these points are collinear.
Then, the condition $\Q(\Fq) \subset \Q'(\Fq)$ forces $\Q=\Q'$.

\medbreak

\noindent \textbf{If $\mathbf{q=3}$}, there are $4$ rational points on $\Q$ such that no $3$ are
collinear.  Hence they impose $4$ independent conditions on
conics~\eqref{eq:conic}. Thus the set of conics interpolating
$\Q(\F_3)$ has dimension $1$: it is a pencil. Thus, there are $q+1=4$
conics passing through $\Q (\F_3)$. Given these four points such that
no $3$ of them are collinear, there are exactly three ways to choose a
pair of lines (hyperplanes in $\PP^2$) passing through them. See Figure \ref{fig:qis3} for an illustration. This entails that $\Q$ is the unique irreducible conic in this pencil and hence that $\Q'=\Q$.

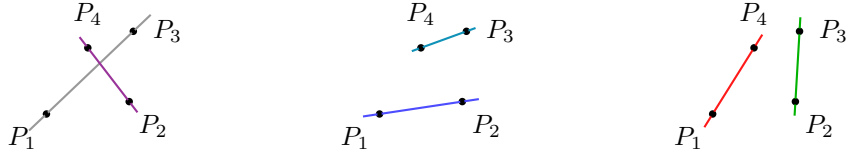
\begin{figure}[h]
\centering
\begin{minipage}{0.3\textwidth}
\centering
\begin{tikzpicture}[scale=0.55, every node/.style={circle, fill=black, inner sep=1pt, label distance=1mm}]
\node (A) at (0,0) [label=below left:$P_1$] {};
\node (B) at (2,0.3) [label=below right:$P_2$] {};
\node (C) at (2.1,2) [label=right:$P_3$] {};
\node (D) at (1,1.6) [label=above:$P_4$] {};
\foreach \P/\Q in {A/C} {
  \draw[thick, gray!80] ($( \P )! -0.2! ( \Q )$) -- ($( \P )! 1.2! ( \Q )$);
}
\foreach \P/\Q in {B/D} {
  \draw[thick, violet!80] ($( \P )! -0.2! ( \Q )$) -- ($( \P )! 1.2! ( \Q )$);
}
\foreach \P in {A,B,C,D} \fill (\P) circle (1.5pt);
\end{tikzpicture}
\end{minipage}
\hfill
\begin{minipage}{0.3\textwidth}
\centering
\begin{tikzpicture}[scale=0.55, every node/.style={circle, fill=black, inner sep=1pt, label distance=1mm}]
\node (A) at (0,0) [label=below left:$P_1$] {};
\node (B) at (2,0.3) [label=below right:$P_2$] {};
\node (C) at (2.1,2) [label=right:$P_3$] {};
\node (D) at (1,1.6) [label=above:$P_4$] {};
\foreach \P/\Q in {A/B} {
  \draw[thick, blue!70] ($( \P )! -0.2! ( \Q )$) -- ($( \P )! 1.2! ( \Q )$);
}
\foreach \P/\Q in {C/D} {
  \draw[thick, cyan!70!black] ($( \P )! -0.2! ( \Q )$) -- ($( \P )! 1.2! ( \Q )$);
}
\foreach \P in {A,B,C,D} \fill (\P) circle (1.5pt);
\end{tikzpicture}
\end{minipage}
\hfill
\begin{minipage}{0.3\textwidth}
\centering
\begin{tikzpicture}[scale=0.55, every node/.style={circle, fill=black, inner sep=1pt, label distance=1mm}]
\node (A) at (0,0) [label=below left:$P_1$] {};
\node (B) at (2,0.3) [label=below right:$P_2$] {};
\node (C) at (2.1,2) [label=right:$P_3$] {};
\node (D) at (1,1.6) [label=above:$P_4$] {};
\foreach \P/\Q in {A/D} {
  \draw[thick, red!90] ($( \P )! -0.2! ( \Q )$) -- ($( \P )! 1.2! ( \Q )$);
}
\foreach \P/\Q in {B/C} {
  \draw[thick, green!70!black] ($( \P )! -0.2! ( \Q )$) -- ($( \P )! 1.2! ( \Q )$);
}
\foreach \P in {A,B,C,D} \fill (\P) circle (2pt);
\end{tikzpicture}
\end{minipage}

\caption{Reducible conics through $P_1,P_2,P_3,P_4$ with no three collinear.}
\label{fig:qis3}
\end{figure}

\noindent \textbf{If $\mathbf{q=2}$}, then there are $3$ non collinear rational
points on $\Q$. Hence, the linear system interpolating these points is
parameterized by $\PP^2$: there are $q^2+q+1=7$
rational conics passing through these $3$ points. Let us count the
$\Fq$-reducible conics again. Denote by $P_1, P_2, P_3$ the $3$
rational points and by $\ell_{ij}$ the line passing through $P_i$ and
$P_j$.  Take line $\ell_{12}$ passing through $P_1, P_2$ and an arbitrary
rational line passing through $P_3$. There are $3$ possible choices
for the line through $P_3$, which leads to $3$ reducible conics
passing through $P_1, P_2, P_3$ and containing $\ell_{12}$. Proceed to the same operation
considering the union of $\ell_{13}$ with a line passing through
$P_2$. Similarly there are $3$ such reducible conics including
$\ell_{13} \cup \ell_{12}$ which was already counted. This yields $2$
additional conics.  Finally consider the union of $\ell_{23}$ with a line
through $P_1$: we already counted $\ell_{23}\cup \ell_{12}$ and
$\ell_{23} \cup \ell_{13}$ and there remains the line passing through $P_1$
and the third rational point of $\ell_{23}$. This yields one additional
reducible conic. In summary, there are $6$ reducible conics in this set
of seven conics interpolating $\Q(\Fq) = \{P_1, P_2, P_3\}$. See Figure~\ref{fig:qis2} for an illustration. Hence, as
in the $\F_3$ case, this entails that $\Q$ is the unique irreducible
conic in this set and hence that $\Q = \Q'$. 

\begin{figure}[h]
\centering
\begin{minipage}{0.3\textwidth}
\centering
\begin{tikzpicture}[scale=0.8]
\node (A) at (0,0) [label=below left:$P_1$] {};
\node (B) at (1.6,0.2) [label=below right:$P_3$] {};
\node (C) at (1.2,1.2) [label=above:$P_2$] {};
\node (D) at (1.8,1.8) [label=above:${}$] {};
\foreach \P/\Q in {A/D} {
  \draw[thick, blue!80] ($( \P )! -0.2! ( \Q )$) -- ($( \P )! 1.2! ( \Q )$)
  node [midway, left] {$\ell_{12}$};
}
\foreach \P/\Q in {B/A} {
  \draw[thick, gray!80] ($( \P )! -0.2! ( \Q )$) -- ($( \P )! 1.2! ( \Q )$)
  node [midway, below] {$\ell_{13}$};
}
\foreach \P in {A,B,C,D} \fill (\P) circle (1.5pt);

\draw[red, thick] (D) -- ++(-0.1,-0.1) -- ++(0.2,0.2)  ++(-0.2,0) -- ++(0.2,-0.2);
\end{tikzpicture}
\end{minipage}
\hfill
\begin{minipage}{0.3\textwidth}
\centering
\begin{tikzpicture}[scale=0.8]
\node (A) at (0,0) [label=below left:$P_1$] {};
\node (B) at (1.6,0.2) [label=below right:$P_3$] {};
\node (C) at (1.2,1.2) [label=above:$P_2$] {};
\node (D) at (1.8,1.8) [label=above:${}$] {};
\foreach \P/\Q in {A/D} {
  \draw[thick, blue!80] ($( \P )! -0.2! ( \Q )$) -- ($( \P )! 1.2! ( \Q )$)
  node [midway, left] {$\ell_{12}$};
}
\foreach \P/\Q in {B/C} {
  \draw[thick, red!70!black] ($( \P )! -0.2! ( \Q )$) -- ($( \P )! 1.2! ( \Q )$)
  node [midway, right] {$\ell_{23}$};
}
\foreach \P in {A,B,C,D} \fill (\P) circle (1.5pt);

\draw[red, thick] (D) -- ++(-0.1,-0.1) -- ++(0.2,0.2)  ++(-0.2,0) -- ++(0.2,-0.2);
\end{tikzpicture}
\end{minipage}
\hfill
\begin{minipage}{0.3\textwidth}
\centering
\begin{tikzpicture}[scale=0.8]
\node (A) at (0,0) [label=below left:$P_1$] {};
\node (B) at (1.6,0.2) [label=below right:$P_3$] {};
\node (C) at (1.2,1.2) [label=above:$P_2$] {};
\node (D) at (1.8,1.8) [label=above:${}$] {};
\foreach \P/\Q in {A/D} {
  \draw[thick, blue!80] ($( \P )! -0.2! ( \Q )$) -- ($( \P )! 1.2! ( \Q )$)
  node [midway, left] {$\ell_{12}$};
}
\foreach \P/\Q in {B/D} {
  \draw[thick, magenta!70] ($( \P )! -0.2! ( \Q )$) -- ($( \P )! 1.2! ( \Q )$)
  ;
}
\foreach \P in {A,B,C,D} \fill (\P) circle (1.5pt);

\draw[red, thick] (D) -- ++(-0.1,-0.1) -- ++(0.2,0.2)  ++(-0.2,0) -- ++(0.2,-0.2);
\end{tikzpicture}
\end{minipage}

\centering
\begin{minipage}{0.3\textwidth}
\centering
\begin{tikzpicture}[scale=1]
\node (A) at (0,0) [label=below left:$P_1$] {};
\node (B) at (1.6,0.2) [label=below right:$P_3$] {};
\node (C) at (1.2,1.2) [label=left:$P_2$] {};
\node (E) at (2.5,0.3125) [label=above:${}$] {};
\foreach \P/\Q in {A/E} {
  \draw[thick, green!70!black] ($( \P )! -0.2! ( \Q )$) -- ($( \P )! 1.2! ( \Q )$)
  node [midway, below] {$\ell_{13}$};
}
\foreach \P/\Q in {B/C} {
  \draw[thick, red!70!black] ($( \P )! -0.2! ( \Q )$) -- ($( \P )! 1.2! ( \Q )$)
  node [midway, right] {$\ell_{23}$};
}
\foreach \P in {A,B,C,E} \fill (\P) circle (1.5pt);

\draw[red, thick] (E) -- ++(-0.1,-0.1) -- ++(0.2,0.2)  ++(-0.2,0) -- ++(0.2,-0.2);
\end{tikzpicture}
\end{minipage}
\hfill
\begin{minipage}{0.3\textwidth}
\centering
\begin{tikzpicture}[scale=0.8]
\node (A) at (0,0) [label=below left:$P_1$] {};
\node (B) at (1.6,0.2) [label=below right:$P_3$] {};
\node (C) at (1.2,1.2) [label=left:$P_2$] {};
\node (E) at (2.5,0.3125) [label=above:${}$] {};
\foreach \P/\Q in { A/E} {
  \draw[thick, green!70!black] ($( \P )! -0.2! ( \Q )$) -- ($( \P )! 1.2! ( \Q )$)
  node [midway, below] {$\ell_{13}$};
}
\foreach \P/\Q in {C/E} {
  \draw[thick, violet!80] ($( \P )! -0.2! ( \Q )$) -- ($( \P )! 1.2! ( \Q )$);
}
\foreach \P in {A,B,C,E} \fill (\P) circle (1.5pt);

\draw[red, thick] (E) -- ++(-0.1,-0.1) -- ++(0.2,0.2)  ++(-0.2,0) -- ++(0.2,-0.2);
\end{tikzpicture}
\end{minipage}
\hfill
\begin{minipage}{0.3\textwidth}
\centering
\begin{tikzpicture}[scale=0.8]
\node (A) at (0,0) [label=below left:$P_1$] {};
\node (B) at (1.6,0.2) [label=below right:$P_3$] {};
\node (C) at (1.2,1.2) [label=right:$P_2$] {};
\node (F) at (1,1.7) [label=above:${}$] {};
\foreach \P/\Q in {B/F} {
  \draw[thick, red!70!black] ($( \P )! -0.2! ( \Q )$) -- ($( \P )! 1.2! ( \Q )$)
  node [midway, below right] {$\ell_{23}$};
}
\foreach \P/\Q in {A/F} {
  \draw[thick, cyan!80!black] ($( \P )! -0.2! ( \Q )$) -- ($( \P )! 1.2! ( \Q )$);
}
\foreach \P in {A,B,C,F} \fill (\P) circle (1.5pt);

\draw[red, thick] (F) -- ++(-0.1,-0.1) -- ++(0.2,0.2)  ++(-0.2,0) -- ++(0.2,-0.2);
\end{tikzpicture}
\end{minipage}

\caption{Reducible conics through non collinear points $P_1,P_2,P_3$.}
\label{fig:qis2}
\end{figure}
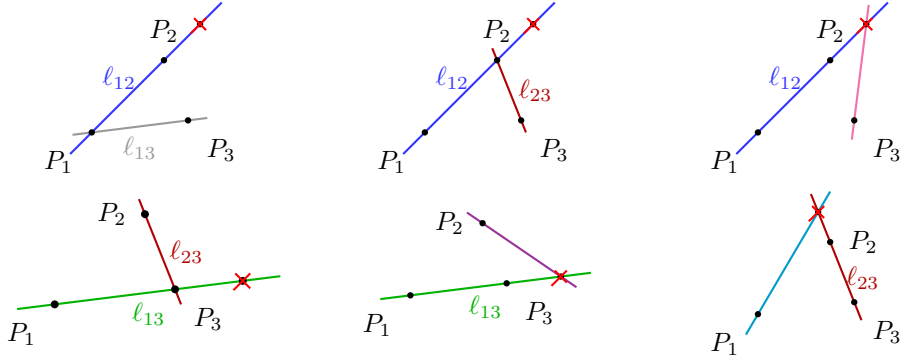

\subsubsection{The case $N=3$}

Recall that, using Proposition~\ref{prop:SingPts} and
Proposition~\ref{prop:rankSing}, we have $r\geq r'$. Moreover, since
$\Q$ is smooth, $r=4$ and, since $\Q'$ is absolutely irreducible, by Proposition~\ref{prop:Rk}, we obtain
$r'\ge 3$. In summary: $4=r\ge r'\ge 3$.

\medskip

\noindent \textbf{Suppose $\Q\sim \E_4$ and $\Q'\sim \E_4$.}
From Theorem~\ref{thm:IrrClassQuadrics}, after a suitable linear
change of variables, one can assume that
$F(X_0,X_1,X_2,X_3)=f(X_0,X_1)+X_2X_3$ where $f$ is $\Fq$-irreducible.
Then, $\Q\cap \V(X_0)$ is parabolic of rank $3$ and, from
Proposition~\ref{prop:HyperplaneCutRk}, $\Q'\cap \V(X_0)$ has rank $2$
or $3$. If $\rk (\Q'\cap \V(X_0))=2$, then either it is a union of two
rational lines or a union of two conjugate $\F_{q^2}$--lines. In the
former case, then $g(\Q'\cap \V(X_0))=1$ which is a contradiction to
$\Q'$ being elliptic of rank $4$ as $g(\Q')=0$ by
Proposition~\ref{prop:ProjIndex}. In the latter case, then
$\Q'\cap \V(X_0)$ has a unique rational point while containing the
$q+1$ rational points of $\Q \cap \V (X_0)$, a contradiction. Thus,
$\Q'\cap \V(X_0)$ is also parabolic of rank $3$. By the result proved
for $N=2$, $\Q\cap \V(X_0)=\Q'\cap \V(X_0)$, \emph{i.e.}
$F'(0,X_1,X_2,X_3)=\lambda F(0,X_1,X_2,X_3)$. Write
$F'=F'(0,X_1,X_2,X_3)+X_0 L'$ for some linear form
$L'\in \Fq[X_0,X_1,X_2,X_3]$. Then,
\begin{equation}\label{eq:1st_difference}
  F'-\lambda F=X_0 (L'-\lambda (X_0+\alpha X_1)),
\end{equation}
which implies that $\Q(\Fq)\subset \V(X_0(L'-\lambda(X_0+ \alpha X_1)))$. By Theorem~\ref{thm:irredandred}, when $q\ne 2$ it follows that $L'=\lambda(X_0+ \alpha X_1)$  and hence $\Q' = \Q$.
When $q=2$, applying the same argument to the hyperplane sections given by $\V(X_1)$, we obtain 
\begin{equation}\label{eq:2nd_difference}
  F'-\lambda'F=X_1(L''-\lambda' (\alpha X_0+X_1))
\end{equation}
for some linear form $L''\in \F_2[X_0,\ldots,X_3]$. Note that, since
we work over $\F_2$ then $\alpha = \lambda = \lambda'=1$.
Taking (\ref{eq:1st_difference}) and (\ref{eq:2nd_difference}) together
shows that the quadratic form $F - F'$ is divisible by both $X_0$ and $X_1$.
Therefore, either $F=F'$ and we are done or $F-F' = X_0X_1$. In the latter case,
$F' = X_0^2+X_1^2 + X_2X_3 = (X_0+X_1)^2+X_2X_3$ would be parabolic of rank $3$,
which is a contradiction. As a conclusion, $\Q = \Q'$.

\medskip

\noindent \textbf{Suppose $r=4$ and $r'=3$.}
Since $\sharp\Q(\Fq)\leq \sharp\Q'(\Fq)$, Proposition~\ref{prop:NoOfPoints} implies that $\Q$ is elliptic. As $\Q'$ being parabolic is a cone with a unique singular point, say $P$, we can describe $\Q'(\Fq)$  as a union of $q+1$ lines through $P$, say $\ell_1,\ldots,\ell_{q+1}$ corresponding to the fact that $\sharp \P_2(\Fq)=q+1$ in $\PP^2$. Then, 
\begin{equation}\label{eq:parabola}
     	\Q(\Fq)  \subset \bigcup_{i=1}^{q+1} (\Q\cap \ell_i)(\Fq) \quad  \implies  \quad  \sharp \Q(\Fq)  \leq \sum_{i=1}^{q+1} \sharp (\Q\cap \ell_i)(\Fq) 
     	\leq 2(q+1),
     \end{equation}
where last inequality uses $g(\Q)=0$
     (Proposition~\ref{prop:ProjIndex}) which shows that none of the
     $\ell_i$'s are contained in $\Q$ and hence all intersect $\Q$ at
     most at two rational points. This gives a contradiction for
     $q\ne 2$ as $\sharp\Q(\Fq)=q^2+1$ by
     Proposition~\ref{prop:NoOfPoints}. If $q=2$, for
     \eqref{eq:parabola} to hold, we should have $P\notin \Q(\Fq)$ as
     otherwise $P\in \Q\cap \ell_i$ for all $i$ and hence
     \eqref{eq:parabola} would imply $5 = q^2+1\leq (q+1) +1=4$.
     Moreover, \eqref{eq:parabola} imposes that
     $\Q\cap \ell_i(\F_2)=2$ for at least two lines, say $\ell_1$ and
     $\ell_2$. Now note that $\ell_1$ and $\ell_2$ are
     coplanar by virtue of sharing the vertex $P$, let
     $\Pi\subset \PP^3$ denote the plane they span. Then,
     $\sharp (\Pi\cap \Q)(\F_2)\geq 4$. But $\rk (\Q\cap \Pi)=2$ or
     $3$ by Proposition~\ref{prop:Rk}. In the first case, since
     $g(\Q)=0$, we obtain $\sharp (\Pi\cap \Q)(\F_2)=1$, corresponding
     to the union of $\F_{2^2}$-lines, and in the second case,
     $\Pi\cap \Q$ is parabolic of rank $3$ and hence
     $\sharp (\Pi\cap \Q)(\F_2)=2+1=3$, a contradiction.

\medskip
     
\noindent \textbf{Suppose $\Q\sim \E_4$ and $\Q'\sim \H_4$.}
By the Segre embedding~\eqref{eq:mapTheta}, $\Q'(\Fq)$ is a union of $q+1$ disjoint lines, say $\ell_1,\ldots, \ell_{q+1}$. Then again,
\begin{equation*}
     	\Q(\Fq)  \subset \bigcup_{i=1}^{q+1} (\Q\cap \ell_i)(\Fq)  \quad \implies \quad \sharp \Q(\Fq)  \leq \sum_{i=1}^{q+1} \sharp (\Q\cap \ell_i)(\Fq) 
     	\leq 2(q+1),
     \end{equation*}   
     a contradiction for $q\geq 3$, consistent with the hypothesis of the theorem that $q\ne 2$ for this case.

\medskip
     
\noindent \textbf{Suppose $\Q\sim \H_4$ and $\Q'\sim \H_4$.}
As in the induction step~(\ref{subsec:InductionStep}), assume again
that $\Q$ is defined by $F(X_0,X_1,X_2,X_3) = X_0 X_1 + X_2 X_3 =
0$. Then, $\Q\cap \V(X_0)$ is a union of distinct lines in $\PP^2$ (it
has equation $X_2X_3 = 0$). Since $\Q(\Fq) \subset \Q'(\Fq)$, then, from
Corollary~\ref{cor:linear}, $\Q' \cap \V(X_0)$ is the same union of
lines and hence is also defined by $X_2X_3=0$. Therefore,
$F'(0,X_1,X_2,X_3)=\lambda F(0,X_1,X_2,X_3)$ for some nonzero scalar
$\lambda$ and hence $F' = \lambda F + X_0 L'$ for some linear form
$L'$. If $L' = 0$, we are done. Otherwise, since
$\Q(\Fq)\subset \Q'(\Fq)$, we obtain $\Q(\Fq)\subset \V(X_0L')(\Fq)$,
\emph{i.e.} $\Q(\Fq)$ is contained in a union of hyperplanes, which
contradicts Theorem~\ref{thm:irredandred}, as desired.

\section{Minimal codewords of projective Reed--Muller codes of
degree $2$}\label{sec:minimal_codewords}

Regarding codes, we restate our main result:

\mincodewords*

\begin{proof}
Follows from Theorem~\ref{thm:irredandred} and Theorem~\ref{thm:main}. 
\end{proof}

For $r\geq 1$, let $\Q_r$ be a smooth quadric in $\PP^{r-1}$, in
particular, $\rk \Q_r=r$. Denote by $\N(\Q_r)$ the number of quadrics
which are projectively equivalent to $\Q_r$. Here we recall the
following statement.

\begin{thm}[{\cite[Thm.~1.44(ii)]{hirschfeld1991general}}]
  We have,
  \begin{itemize}
\item $\N(\P_r)=\displaystyle q^{(r-1)(r+1)/4}\prod_{i=1}^{(r-1)/2}(q^{2i+1}-1)$, \hfill (r odd)
\item $\N(\H_r)=\displaystyle \frac{1}{2}q^{r^2/4}(q^{r/2}+1) \prod_{i=1}^{(r-2)/2} (q^{2i+1}-1)$, \hfill (r even)
\item $\N(\E_r)=\displaystyle \frac{1}{2} q^{r^2/4} (q^{r/2}-1) \prod_{i=1}^{(r-2)/2} (q^{2i+1}-1)$. \hfill (r even)
\end{itemize}
\end{thm}

Using this result, we give the number of minimal codewords of $\prm_q(2,N)$ as well.
Recall that the Gaussian binomial coefficient
\[
  \GaussBinom{n}{k}
\]
denotes the number of $k$--dimensional subspaces of $\F_q^n$.

\begin{thm}
  The number of minimal codewords of weight $w$ in $\prm_q(2,N)$  is
  \begin{enumerate}[$(i)$]
  \item\label{item:i} $(q-1)\GaussBinom{N+1}{2} {q+1 \choose 2},$ if  $w=q^{N}-q^{N-1}$;
    \medskip
  \item\label{item:ii} $(q-1)\sum\limits_{\substack{r=3+\delta \\ r \text{ odd}}}^{N+1}\GaussBinom{N+1}{r}\N(\P_r),$ if $w=q^N$;
    \medskip
  \item\label{item:iii} $(q-1)\GaussBinom{N+1}{r} \N(\H_r)$ if  $w=q^N-q^{N-r/2}$  for $4 \leq r\leq N+1$ even;
    \medskip
  \item\label{item:iv} $(q-1)\GaussBinom{N+1}{r} \N(\E_r)$ if  $w=q^N+q^{N-r/2}$ for  $4+\epsilon \leq r\leq N+1$ even;
    \medskip
  \item\label{item:v} $0$ otherwise,
  \end{enumerate}  
where $\delta=0$ if $q>3$, $\delta=2$ if $q\leq 3$, $\epsilon=0$ if $q\ne 2$ and $\epsilon=2$ if $q=2$. 
\end{thm}
\begin{proof}
  The number of quadratic forms
  corresponding to a fixed quadric is $(q-1)$, which explains the
  $(q-1)$ factor in all the formulas. Using Theorem~\ref{thm:main}, we
  thus require to count the quadrics which are either of the form
  $\P_r, \H_r, \E_r$ for some $r\geq 3$ or a union of distinct rational
  hyperplanes.

   We proceed to count quadrics using
   Proposition~\ref{prop:SingCone}. For a fixed rank $r$ and hence a
   fixed dimension of the singular locus $N-r$, we first count the
   number of possible choices for $\sing (\Q)$ which is nothing but
  \[
    \GaussBinom{N+1}{N+1-r} = \GaussBinom{N+1}{r},
  \]
  where the equality is a well--known property of Gaussian binomial
  coefficients. Once the singular locus is fixed, we chose an
  arbitrary $\Pi$ of dimension $r-1$ and count the number of
  candidates for $\Q \cap \Pi$, which is the number of possible
  nonsingular quadrics of a given class in $\PP^{r-1}$.

  \medskip
  \noindent \emph{Proof of (\ref{item:i}).}  Note that the
  characterisation as a cone holds even for union of hyperplanes
  $\L_1 \cup \L_2$. In the latter case, the singular locus is the
  codimension $2$ linear space at the intersection of the two
  hyperplanes and $\Pi$ can be taken as a line $\ell$ disjoint from
  the aforementioned singular locus. The ``quadric'' $\ell \cap \Q$ is
  then a ``quadric'' in $\PP^1$ which is a union of two points: two
  rational points give rise to a union of rational hyperplanes, while
  two conjugated $\F_{q^2}$--points give rise to an irreducible non
  absolutely irreducible quadric. This yields (\ref{item:i}) : the
  Gaussian binomial term counting the possible choices for the
  singular locus, the usual binomial term counting the number of
  possible unordered pairs of elements in $\PP^{1}(\Fq)$. The $(q-1)$
  factor was explained in the beginning of this proof.

  \medskip
  \noindent \emph{Proof of (\ref{item:ii}).} According to Proposition~\ref{prop:NoOfPoints},
  parabolic quadrics have the same number of rational points whatever
  their rank.  Hence we need to consider all the possible (odd) ranks,
  which explains the summation.  The Gaussian binomial term is the
  number of possible choices for the singular locus and the
  $\N (\P_r)$ is the number of choices for $\Q \cap \Pi$ as explained
  above.  Regarding the summation indexes, we consider all the odd
  integers from $3$ to $N+1$ unless $q=2,3$ where rank $3$ should be
  discarded according to
  Theorem~\ref{thm:intro:mincodewords}(\ref{item:mincodewords2}).

  \medskip
  \noindent \emph{Proof of (\ref{item:iii}), (\ref{item:iv}) and
    (\ref{item:v}).} Here, each class of quadric and each rank gives rise
  to a different number of rational points. Formulas (\ref{item:iii}),
  (\ref{item:iv}) are obtained in the same way as the previous
  ones. Formula (\ref{item:iv}) excludes elliptic quadrics of rank $4$ when $q=2$
  according to
  Theorem~\ref{thm:intro:mincodewords}(\ref{item:mincodewords2}).
\end{proof}

\bibliographystyle{plain}

\end{document}